\documentclass[11pt,a4paper,english]{amsart}
\usepackage{babel}
\usepackage[latin1]{inputenc}
\usepackage{amsmath}
\usepackage{amsthm}
\usepackage{amsfonts}
\usepackage{indentfirst}
\usepackage{graphicx}
\usepackage{amssymb}
\usepackage[mathscr]{eucal}
\usepackage{tikz}
\usepackage{caption}
\usepackage{amsthm}
\usepackage{amscd}
\newtheorem{pro}{Proposition}
\newtheorem{teo}{Theorem}
\newtheorem{rem}{Remark}
\newtheorem{exa}{Example}

\title[Stabilizer quantum codes from $J$-affine variety codes]{Stabilizer quantum codes from $J$-affine variety codes and a new Steane-like enlargement}
\author{Carlos Galindo, Fernando Hernando and Diego Ruano}
\curraddr{\texttt{Carlos Galindo and Fernando Hernando:} Instituto
Universitario de Matem\'aticas y Aplicaciones de Castell\'on and
Departamento de Matem\'aticas, Universitat Jaume I, Campus de Riu
Sec. 12071 Castell\'{o} (Spain)\\
\texttt{Diego Ruano:} Department of Mathematical Sciences, Aalborg University, Fredrik Bajers Vej 7G, 9220 Aalborg East (Denmark).
}
\email{{\rm Galindo:} galindo@uji.es; {\rm Hernando:} carrillf@uji.es; {\rm Ruano:} diego@math.aau.dk}
\date{}
\thanks{Supported by the Spanish Ministry of Economy: grant MTM2012-36917-C03-03, the University Jaume I: grant PB1-1B2012-04, the Danish Council for Independent Research, grant DFF-4002-00367 and the ``Program for Promoting the Enhancement of Research
Universities'' at Tokyo Institute of Technology.}
\keywords{Stabilizer $J$-affine variety codes; Subfield-subcodes; Steane enlargement; Hermitian and Euclidean duality}

\begin{document}

\begin{abstract}
New stabilizer codes with parameters better than the ones available in the literature are provided in this work, in particular quantum codes with parameters $[[127,63, \geq 12]]_2$ and $[[63,45, \geq 6]]_4$ that are records. These codes are constructed  with a new generalization of the Steane's enlargement procedure and by considering orthogonal subfield-subcodes --with respect to the Euclidean and Hermitian inner product-- of a new family of linear codes, the $J$-affine variety codes.
\end{abstract}

\maketitle

\section*{Introduction and preliminaries}

Polynomial time algorithms for prime factorization and discrete logarithms on quantum computers were given by Shor in 1994 \cite{22RBC}.  Thus, if an efficient quantum computer existed (see  \cite{BCM,SS}, for recent advances), most popular cryptographic systems could be broken and much computational work could be done much faster. Unlike classical information, quantum information cannot be cloned \cite{8AS, 26RBC}, despite this fact quantum (error-correcting) codes do exist \cite{23RBC, 95kkk}. The above facts explain why, in the last decades, the interest in quantum computations and, in particular, in quantum coding theory grew dramatically. There exists an extensive literature on quantum codes, see for instance  \cite{7kkk, 8kkk, 18kkk, 19kkk, 20kkk, 38kkk, 45kkk} for the binary case and \cite{AK, BE, 35kkk, opt, lag3, 71kkk} for the general case.

Classical linear codes and Hermitian and Euclidean inner products are useful tools to construct codes in a class of quantum codes named stabilizer codes. In this paper, we introduce $J$-affine variety codes and use them, together with the Hamada's generalization \cite{ham} of the Steane's enlargement procedure, to derive new stabilizer codes. The main procedures in the literature we use are collected in Theorems \ref{bueno} and \ref{ham}. Furthermore, we  introduce and consider a new enlargement that we state in Theorem \ref{elnuevo}. This result extends the mentioned Hamada's generalization and is proved in the appendix. With the above ideas, we obtain binary stabilizer codes which are records in \cite{codet} and nonbinary codes that exceed the Gilbert-Varshamov bounds.

Set $q=p^r$ a positive power of a prime number $p$ and let $\mathbb{C}$ be the complex numbers. A {\it stabilizer code} $\mathcal{C} \neq \{0\}$ is the common eigenspace of an abelian subgroup $\Delta$ of the error group $G_n$ generated by a nice error basis on the space $\mathbb{C}^{q^n}$, $n$ being a positive integer. The code $\mathcal{C}$ has minimum distance $d$ whenever all error in $G_n$ with weight less than $d$ can be detected or have no effect on $\mathcal{C}$ but some error of weight $d$ cannot be detected. In addition,  $\mathcal{C}$ is called to be {\it pure} if $\Delta$ has  not non-scalar matrices with weight less than $d$. Finally, a code as above is an $[[n,k,d]]_q$-code when it is a $q^k$-dimensional subspace of $\mathbb{C}^{q^n}$ and has minimum distance $d$ (see for instance \cite{19kkk,kkk}). Recall that the Hermitian inner product of two vectors $\mathbf{x} =(x_1,x_2, \ldots, x_n)$ and $\mathbf{y}=(y_1,y_2, \ldots, y_n)$ in the vector space $\mathbb{F}_{q^2}^n$ is defined as $\mathbf{x}  \cdot_h \mathbf{y}= \sum x_i y_i^q$ and the Euclidean product of $\mathbf{x}$ and $\mathbf{y}$ in  $\mathbb{F}_{q}^n$ as $\mathbf{x} \cdot \mathbf{y} = \sum x_i y_i$.  Given a linear code $C$ in $\mathbb{F}_{q^2}^n$ (respectively, $\mathbb{F}_{q}^n$), the Hermitian (respectively, Euclidean) dual space is denoted by $C^{\perp_h}$ (respectively, $C^\perp$).

\begin{teo}
\label{bueno}
\cite{kkk,Akk} The following two statements hold.
\begin{enumerate}
    \item Let $C$ be a linear $[n,k,d]$ error-correcting  code over  $\mathbb{F}_{q}$ such that $C^\perp \subseteq C$. Then, there exists an $[[n, 2k -n, \geq d]]_q$ stabilizer code which is pure to $d$. If the minimum distance of $C^\perp$ exceeds $d$, then the stabilizer code is pure and has minimum distance $d$.
        \item Let $C$ be a linear $[n,k,d]$ error-correcting  code over  $\mathbb{F}_{q^2}$ such that $C^{\perp_h }\subseteq C$. Then, there exists an $[[n, 2k -n, \geq d]]_q$ stabilizer code which is pure to $d$. If the minimum distance $d^{\perp_h}$ of the code $C^{\perp_h}$ exceeds $d$, then the stabilizer code is pure and has minimum distance $d$.
\end{enumerate}

\end{teo}

Codes obtained as described in Item (1) of Theorem \ref{bueno} are usually referred to as obtained from the CSS construction \cite{20kkk,95kkk}. This last procedure is not only useful for quantum error-correction but also for privacy amplification of quantum cryptography \cite{3hagiw}. In addition, it has been extended to construct asymmetric quantum codes which are suitable for quantum mechanical systems where the phase-flip errors happen more frequently than the bit-flip errors \cite{sar1, ez}. The parameters of the codes coming from Item (1) of Theorem \ref{bueno}  can be improved with the Hamada's generalization \cite{ham} of the Steane's enlargement procedure \cite{Steane-E}. Let us state the result, where $\mathrm{wt}$ denotes minimum weight.
\begin{teo} \label{ham} \cite{ham}
Let $C$ be an $[n,k]$ linear code over the field $\mathbb{F}_q$ such that $C^\perp \subseteq C$. Assume that $C$ can be enlarged to an $[n,k']$ linear code $C'$, where $k' \geq k + 2$. Then, there exists a stabilizer code with parameters $[[n,k+k'-n, d \geq \min \{d', \lceil \frac{q+1}{q} d" \rceil\} ]]_q$, where $d' = \mathrm{wt} (C \setminus C'^\perp)$ and $d" = \mathrm{wt} (C' \setminus C'^\perp)$.
\end{teo}
Using the above results, many quantum codes coming from classical codes have been constructed, see \cite{kkk} for example. A complete table of parameters corresponding to known binary quantum codes, up to length 128, can be consulted in \cite{codet}. There is no table for non-binary quantum codes, although there are some codes with good parameters, essentially concerning MDS quantum codes, quantum LDPC codes or quantum BCH codes  \cite{Sarvepalli, edel, Akk, lag3, lag1, jin, lag2, f2}. Most of the above mentioned codes have specific lengths depending on $q$. A way to get codes with good parameters uses evaluation (classical) codes \cite{geil,Geil-Affine,galmon,geil2,galmon2}. In \cite{galindo-hernando, gal-her-rua} the authors consider affine variety codes which form a class of evaluation codes such that duality can be characterized. The reader can consult \cite{Geil-Affine} for a lower bound for the minimum distance of these codes.

The above mentioned papers \cite{galindo-hernando, gal-her-rua} considered affine variety codes, where we could compare parameters of our codes with others given by BCH codes and improve some of them. The evaluation at zero was not  considered  and we only used duality with respect to the Euclidean inner product.  This paper is devoted to construct algebraically generated stabilizer codes from a more general version of affine variety codes ($J$-affine variety codes), where we decide which coordinates of the points to evaluate may be zero. In this way, we get a wider range of lengths for our codes. Moreover, in this work, both Euclidean and Hermitian duality are considered, which allows us to obtain a richer family of codes. Notice that our codes are based on cyclotomic sets and subfield-subcodes, so they can be seen as a natural extension of BCH codes since these can be constructed with cyclotomic cosets and are subfield-subcodes of Reed-Solomon codes.

Stabilizer codes derived from the Euclidean inner product and  $J$-affine variety codes (and their subfield-subcodes) are studied in Section 1 and their parameters are described in Theorem \ref{encero}. Section 2 develops the Hermitian case, the main result is Theorem \ref{esteo6}, which also gives the parameters of the corresponding codes. To prove this result we show that  Delsarte Theorem \cite{delsarte} also holds in our case, that is, with respect to Hermitian inner product. Furthermore, we prove in the appendix  the following generalization of the Steane's enlargement procedure.

\begin{teo}
\label{elnuevo}
Let $C_1$ and $\hat{C}_1$ be two linear codes over the field $\mathbb{F}_q$, with parameters $[n,k_1,d_1]$ and $[n,\hat{k}_1,\hat{d}_1]$ respectively, and such that $C_1^\perp \subseteq \hat{C}_1$. Consider a linear code $D \subseteq  \mathbb{F}_q^n$  such that $\dim D \ge 2$ and $(C_1 + \hat{C}_1) \cap D = \{ 0 \}$. Set $C_2 = C_1 + D$ and $\hat{C}_2 = \hat{C}_1 + D$, that enlarge $C_1$ and $\hat{C}_1$  respectively, with parameters  $[n,k_2,d_2]$ and $[n,\hat{k}_2,\hat{d}_2]$ ($k_2 - k_1 = \hat{k}_2 - \hat{k_1} = \dim D >1$). Set $C_3$ the code sum of the vector spaces $C_1 + \hat{C}_1 + D$, whose parameters we denote by $[n,k_3,d_3]$. Then, there exists a stabilizer code with parameters
\[
\left[\left[ n,  k_2 + \hat{k}_1-n, d \geq \min \left\{ d_1, \hat{d}_1, \left\lceil \frac{d_2 + \hat{d}_2 + d_3}{2} \right\rceil \right\} \right]\right]_2,
\]
when $q=2$. Otherwise, the parameters are $$ \left[\left[ n,  k_2 + \hat{k}_1-n, d \geq \min \left\{ d_1, \hat{d}_1, M \right\} \right]\right]_q,$$ where $M = \max \{ d_3 + \lceil (d_2/q) \rceil, d_3 + \lceil (\hat{d}_2/q) \rceil\}$.
\end{teo}

In Section \ref{ejemplos} (see Table \ref{ta:best}), we use Theorem \ref{elnuevo} to determine  stabilizer binary codes of length 127 which are records in \cite{codet}. Within the same section and with the help of the previously mentioned results, we provide tables with unknown stabilizer codes over different ground fields that exceed the Gilbert-Varshamov bounds \cite{eck, mat, feng}, \cite[Lemma 31]{kkk}. When comparisons are possible, our codes improve those available in the literature.

\section{Stabilizer $J$-affine variety codes: Euclidean inner product}
\label{laprimera}

In this section we introduce $J$-affine variety codes and characterize their duality with respect to the Euclidean inner product. Our results provide stabilizer quantum codes, derived from these codes and their subfield-subcodes, whose parameters are also described.
\subsection{Euclidean duality for $J$-affine variety codes}
\label{SS1}
Set $q=p^r$ a positive power of a prime number $p$ and consider the finite field $\mathbb{F}_q$. Next we are going to introduce a family of affine variety codes and study their dual codes.

Consider the ring of polynomials  $\mathbb{F}_{q}[X_1, X_2, \ldots, X_m]$ in $m$ variables over the field $\mathbb{F}_{q}$ and fix $m$ integers $N_j >1$ such that  $N_j -1$ divides $q-1$ for $1 \leq j \leq m$.
For a subset $J \subseteq \{1,2, \ldots, m\}$, set $I_J$ the ideal of the ring $\mathbb{F}_{q}[X_1, X_2, \ldots, X_m]$ generated by $X_j^{N_j}-X_j$ whenever $j\not \in J$ and by $X_j^{N_j -1}- 1$ otherwise, for $1 \leq j \leq m$. We denote by $R_J$ the quotient ring
\[
R_J := \mathbb{F}_{q}[X_1, X_2, \ldots, X_m] / I_J.
\]

Set $Z_J = Z(I_J) = \{P_1, P_2, \ldots, P_{n_J}\}$ the set of zeros over $\mathbb{F}_q$ of the defining ideal of $R_J$. Clearly, the points $P_i$, $1 \leq i \leq n_J$, can have $0$ as a coordinate for those indices $j$ which are not in $J$ but this is not the case for the remaining coordinates. Denote  $\mathrm{ev}_J: R_J \rightarrow \mathbb{F}_{q}^{n_J}$ the evaluation map defined as $\mathrm{ev}_J(f) = (f(P_1), f(P_2), \ldots, f(P_{n_J}) $, where $n_J = \prod_{j \notin J} N_j \prod_{j \in J} (N_j -1)$. Denote also $T_j = N_j -1$ except when $j \in J$, in this last case, $T_j = N_j -2$, consider the set
$$
\mathcal{H}_J := \{0,1,\ldots,T_1\}\times \{0,1,\ldots,T_2\} \times\cdots\times\{0,1,\ldots,T_m\}
$$
and a nonempty subset $\Delta \subseteq \mathcal{H}_J$.
Then, we define the  $J$-affine variety code given by $\Delta$, $E_\Delta^J$, as the vector subspace (over $\mathbb{F}_{q}$) of $\mathbb{F}_{q}^{n_J}$ generated by the evaluation by $\mathrm{ev}_J$ of the set of classes in $R_J$ corresponding to monomials $X^{\boldsymbol{a}} := X_1^{a_1} X_1^{a_2}\cdots X_{m}^{a_m}$
such that $\boldsymbol{a}=(a_1,a_2, \ldots, a_m) \in \Delta$. Stabilizer codes constructed from $\{1,2, \ldots, m\}$-affine variety codes were considered in  \cite{galindo-hernando, gal-her-rua}  because they allowed us to do comparisons with some quantum BCH codes. In this paper $\emptyset$-affine variety codes are named  evaluating at zero affine variety codes although in some papers they are simply called affine variety codes \cite{Geil-Affine}.  We will stand $\mathcal{H}$ for $\mathcal{H}_\emptyset$ and we will also write $ \mathcal{H'} := \mathcal{H}_{\{1,2, \ldots, m\}}$. Notice that considering different sets $J$ we get codes of different lengths $$(N_1 -1) (N_2 -1) \cdots (N_m-1)= n_{\{1,2, \ldots, m\}}\le n_J  \le n_\emptyset =N_1 N_2 \cdots N_m.$$

Generalized Reed-Muller codes are a well-known family of evaluating at zero affine variety codes. Indeed, they can be defined as $RM(r,m) := E_{\Delta^0_{(r,m)}}$, where $N_j=q$ for all $j$ and  $\Delta^0_{(r,m)}$ corresponds with the exponents of the monomials in the set $\{f \in R_\emptyset | \deg f \leq r\}$, $\deg f$ meaning the total degree of the unique  representative of $f$ of degree less than $q$ in each indeterminate.

The following result extends one given in \cite{maria-michael} for $N_j=q$, $1 \leq j \leq m$, and it
will be used  for describing dual codes of $J$-affine variety codes.

\begin{pro}
\label{dual}
Let $J \subseteq \{ 1 , 2, \ldots , m\}$, consider $\boldsymbol{a}, \boldsymbol{b} \in \mathcal{H}_J$ and let $X^{\boldsymbol{a}}$ and $X^{\boldsymbol{b}}$ be two monomials representing elements in $R_J$. Then, the Euclidean inner product $\mathrm{ev}_J ( X^{\boldsymbol{a}}) \cdot \mathrm{ev}_J (X^{\boldsymbol{b}})$ is not $0$ if,  and only if, the following two conditions happen.
\begin{itemize}
\item For every $j \in J$, it holds that $a_j + b_j \equiv  0 \mod (N_j -1)$, (i.e.,  $a_j = N_j -1 - b_j$ when $a_j  + b_j > 0$ or $a_j=b_j=0$).
\item For every $j \notin J$, it holds that \begin{itemize}
\item either $a_j  + b_j > 0$ and $a_j + b_j \equiv 0 \mod (N_j -1)$,  (i.e.,  $a_j = N_j -1 - b_j$  if $0 < a_j, b_j < N_j -1$ or $(a_j,b_j) \in \left\{(0,N_j -1), (N_j -1,0), (N_j -1,N_j -1)  \right\}$ otherwise,
\item or $a_j = b_j = 0$ and $p \not | ~ N_j$.
\end{itemize}
\end{itemize}
\end{pro}
\begin{proof}

For $j=1, 2, \ldots, m$, pick an element $\xi_j \in \mathbb{F}_q$ with order $N_j-1$; the existence is guaranteed by the fact that $N_j -1$ divides $q-1$. Then $\langle \xi_j \rangle = \{ \xi_j^0 , \xi_j^1 , \ldots , \xi_j^{N_j-1} \} = Z(X_j^{N_j-1} -1)$ and  $\langle \xi_j \rangle \cup \{ 0\} = Z(X_j^{N_j} -X_j)$. By the distributive law, one has that
$$\mathrm{ev}_J ( X^{\boldsymbol{a}}) \cdot \mathrm{ev}_J (X^{\boldsymbol{b}}) =  \left( \prod_{j \in J} \sum_{\gamma_i \in \langle \xi_j \rangle } \gamma_i^{a_j + b_j} \right) \left(  \prod_{j \not\in J} \sum_{\gamma_i \in \langle \xi_j \rangle \cup \{ 0\}} \gamma_i^{a_j + b_j}\right ).$$ Therefore the previous product is different from zero if, and only if, every factor is different from zero.

Let us consider $j \in J$ and assume that $a_j  + b_j > 0$ and  $a_j = N_j -1 - b_j$, then $a_j + b_j = N_j -1$, and it happens that $\sum_{\gamma_i \in \langle \xi_j \rangle } \gamma_i^{a_j + b_j} \neq 0 $  because$$
\sum_{\gamma_i \in \langle \xi_j \rangle} \gamma_i^{a_j + b_j} = \sum_{\gamma_i \in \langle \xi_j \rangle} \gamma_i^0 = N_j-1 \neq 0 ~(\text{in}~\mathbb{F}_q).
$$
The same result holds for $a_j=b_j=0$. Note that $N_j -1 \neq 0$ in $\mathbb{F}_q$ since $p \not| ~  N_j -1$. Indeed if $p ~| ~ N_j -1$, then, as $N_j -1 | ~ p^r -1$, $p$ had to divide $p^r -1$ which is false. It remains to show what happens when $a_j + b_j \not \equiv 0 \mod (N_j -1)$. In this case $a_j + b_j = c \neq 0$ in  the ring of congruences modulo $N_j -1$, which we set $\mathbb{Z}_{N_j-1}$, and the following chain of equalities holds: $$\sum_{\gamma_i \in \langle \xi_j \rangle} \gamma_i^{a_j + b_j} = \sum_{i=0}^{N_j-2} (\xi_j^i)^{c} = \sum_{i=0}^{N_j-2} (\xi_j^c)^{i} = \frac{1-(\xi_j^c)^{N_j-1}}{1 - \xi_j^c} = 0,$$ which completes the proof for the case $j \in J$. Notice that $\xi_j^c \neq 1$ since $c \neq 0 $ in  $\mathbb{Z}_{N_j-1}$.

To finish, assume $j \not \in J$. We remark that $0^k = 0$ for $k \neq 0$ and $0^0 = 1$.  If $a_j + b_j >0$ then $$\sum_{\gamma_i \in \langle \xi_j \rangle \cup \{ 0\}} \gamma_i^{a_j + b_j}  = \sum_{\gamma_i \in \langle \xi_j \rangle } \gamma_i^{a_j + b_j}$$and the corresponding factor will be different from zero if and only if $a_j + b_j \equiv 0 \mod (N_j -1)$ (by the case $j \in J$). However, if $a_j = b_j =0$ then $$\sum_{\gamma_i \in \langle \xi_j \rangle \cup \{ 0\}} \gamma_i^{a_j + b_j} = 1+ \sum_{\gamma_i \in \langle \xi_j \rangle} \gamma_i^0 = N_j$$ that will be equal to zero if and only if $p ~| ~N_j$.
\end{proof}


The above result shows that each monomial $X^{\boldsymbol{a}} = X_1^{a_1} X_2^{a_2}\cdots X_{m}^{a_m}$, $\boldsymbol{a} \in \mathcal{H}$, admits $2^{\mathrm{card} (Q)}$ monomials $X^{\boldsymbol{b}}$ such that $\mathrm{ev}_J (X^{\boldsymbol{a}}) \cdot \mathrm{ev}_J 	 (X^{\boldsymbol{b}}) \neq 0$, where $$Q= \{j \; | \; 1 \leq j \leq m; a_j = N_j -1\}.$$

Now, for a set $J$ as above, consider a subset $\Delta$ of $\mathcal{H}_J$. If $\Delta \subseteq \mathcal{H'}$, we define $\Delta^\perp$ as the set
$$ \mathcal{H}_J \setminus \{ (N_1 -1 -a_1, N_2 -1 - a_2, \ldots, N_m- 1 -a_m) \; | \; \boldsymbol{a} \in \Delta\}.$$
Otherwise, i.e., in our monomials there is an exponent of some $X_j$ equal to $N_j-1$,  $\Delta^\perp$ is defined as
$$
\mathcal{H}_J\setminus \left\{ \{ (N_1 -1 -a_1, N_2 -1 -a_2, \ldots , N_m -1 -a_m ) | \boldsymbol{a} \in \Delta \cap \mathcal{H'} \} \cup \{ \boldsymbol{a}' | \boldsymbol{a} \in \Delta , \boldsymbol{a} \notin \mathcal{H'} \} \right\},$$ where  $a'_j = N_j -1 -a_j$ if $a_j \neq N_j -1$ and $a'_j$ equals either $N_j -1$ or $0$ otherwise. Notice that an element $\boldsymbol{a} \in \Delta $, $\boldsymbol{a} \not \in \mathcal{H'}$, determines several values $\boldsymbol{a}'$. \\

This definition allows us to state the following straightforward result:
\begin{pro}
\label{pro:str}
Consider a set $\Delta \subseteq \mathcal{H}_J$ as above.
\begin{enumerate}
\item If the following inclusion $\Delta \subseteq \mathcal{H'}$ happens, then the equality of codes $(E^J_\Delta)^\perp = E^J_{\Delta^\perp}$ holds, where $(E^J_\Delta)^\perp$ denotes the dual code of $E^J_\Delta$.
\item Otherwise, if $\Delta \not \subseteq \mathcal{H'}$, then $E^J_{\Delta^\perp} \subseteq (E^J_\Delta)^\perp$.
\end{enumerate}
\end{pro}

\begin{rem}
{\rm
When considering $\{1,2, \ldots,m\}$-affine variety codes, the defining set $\Delta$ must satisfy $\Delta \subseteq \mathcal{H'}$. Thus the same reasoning as above shows that $(E^{\{1,2, \ldots,m\}}_\Delta)^\perp = E^{\{1,2, \ldots,m\}}_{\Delta^\perp}$. 
}
\end{rem}

\subsection{Subfield-subcodes of $J$-affine variety codes}
\label{SS2}
In this section we show some results concerning dimension and self-orthogonality with respect to Euclidean inner product of subfield-subcodes of $J$-affine variety codes.

Recall that $q=p^r$ and pick a positive integer $s$ such that $s$ divides $r$. With the above notations, consider the set $\mathcal{H}_J$ and recall that $N_j -1$ divides $q -1$ for all $j$ such that $1 \leq j \leq m$. Next we define three trace type maps which will be useful: $\mathrm{tr}_r^s: \mathbb{F}_{p^r} \rightarrow \mathbb{F}_{p^s}$ defined as $\mathrm{tr}_r^s(x)= x + x^{p^s} + \cdots + x^{p^{s(\frac{r}{s}-1)}}$; $\mathbf{tr}: \mathbb{F}_{p^{r}}^{n_J} \rightarrow \mathbb{F}_{p^{s}}^{n_J}$, determined by  $\mathrm{tr}_r^s$ componentwise and $\mathcal{T}: R_J \rightarrow R_J$, $\mathcal{T}(f) = f + f^{p^s} + \cdots + f^{p^{s(\frac{r}{s}-1)}}$.

For $1 \leq j \leq m$ consider the above defined integer numbers $T_j$ and, as before, denote by $\mathbb{Z}_{T_j}$ the quotient ring $\mathbb{Z} / T_j \mathbb{Z}$. In this section, we will consider {\it cyclotomic sets} that  is subsets $\mathfrak{I}$ of the cartesian product $\mathbb{Z}_{T_1}\times \mathbb{Z}_{T_2} \times \cdots\times\mathbb{Z}_{T_m}$ such that
$
\mathfrak{I}= \{p^s \cdot \boldsymbol{a} \;| \; \boldsymbol{a}\in \mathfrak{I}\}
$,
where $p^s \cdot \boldsymbol{a} = (p^s a_1, p^s a_2, \ldots, p^s a_m)$. A cyclotomic set $\mathfrak{I}$ is {\it minimal} (for the above given exponent $s$) whenever all the elements in $\mathfrak{I}$ can be expressed as  $p^{s i } \cdot \boldsymbol{a}$ for some fixed element $\boldsymbol{a} \in \mathfrak{I}$ and some nonnegative integer $i$. Consider a set  $\mathcal{A}$ representing the minimal cyclotomic sets, that is  pick $\mathbf{a} \in \mathfrak{I}$ for each minimal cyclotomic set in such a way that $ \mathfrak{I} = \mathfrak{I}_\mathbf{a}$ for some $\mathbf{a} \in \mathcal{A}$. Thus, the set of minimal cyclotomic sets will be $\{ \mathfrak{I}_\mathbf{a}\}_{\mathbf{a} \in \mathcal{A}}$. Moreover, set $i_\mathbf{a} : = \mathrm{card}(\mathfrak{I}_\mathbf{a})$.


The subfield-subcodes (over $\mathbb{F}_{p^s}$) of our $J$-affine variety codes $E^J_\Delta$  are defined as $E^{J,\sigma}_\Delta := E^J_\Delta \cap \mathbb{F}_{p^s}^{n_J}$. We write $C^J_\Delta$ (respectively,
$C^{J,\sigma}_\Delta$) the dual code of $E^{J}_\Delta$ (respectively, $E^{J,\sigma}_\Delta$). Moreover, an {\it element $f \in R_J$  evaluates to $\mathbb{F}_{p^s}$} whenever $f(\boldsymbol{a}) \in \mathbb{F}_{p^s}$ for all $\boldsymbol{a}  \in Z_J$. Notice that this happens if and only if $f = \mathcal{T}(g)$ for some $g \in R_J$. Now we are ready to state the following result that determines the dimension of the subfield-subcodes $E^{J,\sigma}_\Delta $. It can be proved reasoning as in \cite[Theorem 3]{galindo-hernando}.

\begin{teo}
\label{eselteotres}
Let $\beta_\mathbf{a}$ be a primitive element of the finite field $\mathbb{F}_{p^{si_\mathbf{a}}}$ and set $\mathcal{T}_\mathbf{a} : R_J \rightarrow R_J$ the mapping defined as $\mathcal{T}_\mathbf{a}(f) = f + f^{p^s} + \cdots + f^{p^{s(i_\mathbf{a} -1)}}$. Consider a set $\Delta \subseteq \mathcal{H}_J$. Then, the vector space $E^{J,\sigma}_\Delta$ is generated by the images under the evaluation map $\mathrm{ev}_J$ of the following elements in $R_J$:
$\bigcup_{\mathbf{a} \in \mathcal{A}| \mathfrak{I}_\mathbf{a}\subseteq \Delta}  \left\{ \mathcal{T}_{\mathbf{a}} (\beta_\mathbf{a}^{l} X^\mathbf{a}) \; | \; 0 \leq l \leq i_\mathbf{a}-1  \right\}$.
\end{teo}

Next, we provide a result concerning the dimension of the dual code $C^{J,\sigma}_\Delta$.

\begin{teo}
\label{eraeldoce}
Let $\Delta$ be a subset of $\mathcal{H}_J$. Consider the dual code $C_\Delta^{J,\sigma}$ of the subfield-subcode $E_\Delta^{J,\sigma}$. Then:
\begin{enumerate}
  \item The dimension of the code $C_\Delta^{J,\sigma}$ satisfies the inequality
  \[ \dim (C_\Delta^{J,\sigma}) \geq \sum_{\mathbf{a} \in \mathcal{A}| \mathfrak{I}_\mathbf{a} \cap \Delta^\perp \neq \emptyset} i_\mathbf{a}.
\]
  \item If  $\mathfrak{I}_\mathbf{a} \cap \Delta^\perp \neq \emptyset$ whenever $\mathfrak{I}_\mathbf{a} \subseteq \Delta$, then the inclusion $E_\Delta^{J,\sigma} \subseteq C_\Delta^{J,\sigma}$ holds.
  \item Assume that $\Delta$ is a subset of $\mathcal{H'}$. Then we get an equality in {\rm (1)} and the conditions given  in {\rm (2)} are equivalent.
\end{enumerate}
\end{teo}
\begin{proof}
We keep the above notation and recall that $E^J_{\Delta^\perp} \subseteq (E^J_{\Delta})^\perp$. Moreover $(E_\Delta^{J,\sigma})^\perp = \mathbf{tr} (E_\Delta^{J})^\perp $ holds by Delsarte theorem \cite{delsarte}.  Therefore we get $ \mathbf{tr} (E^J_{\Delta^\perp}) \subseteq (E_\Delta^{J,\sigma})^\perp $. Set $\mathbb{F}_{p^r}^{\Delta^\perp}$ the vector space over the field $\mathbb{F}_{p^r}$ of polynomials generated by monomials with exponents in $\Delta^\perp$, which is generated by the set $\{ \mathcal{T} (\gamma X^{\mathbf{a}}\}_{\mathbf{a} \in \Delta^\perp, \gamma \in \mathbb{F}_{p^r}}$. Taking into account that $\mathrm{ev} \circ \mathcal{T} = \mathbf{tr} \circ \mathrm{ev}$,
we deduce that $\mathrm{ev} \left(\mathcal{T} (\mathbb{F}_{p^r}^{\Delta^\perp})\right) = \mathbf{tr} (E^J_{\Delta^\perp})$ which concludes the proof of items (1) and (2). Item (3) follows from the same reasoning and the equality $E^J_{\Delta^\perp} = (E^J_{\Delta})^\perp$.
\end{proof}


\subsection{Results on stabilizer codes}
The results and ideas in Subsections \ref{SS1} and \ref{SS2} together with Theorem \ref{bueno} prove the following result which, keeping the notations as above, states some results for stabilizer codes constructed with $J$-affine variety codes.

\begin{teo}
\label{encero}
Let $N_j$, $1 \leq j \leq m$, be positive integers such that $N_j -1 $ divides $q  -1$ for all index $j$. Let $\Delta$ be a subset of the above defined set $\mathcal{H}_J$. Then:
\begin{enumerate}
\item Assume the set inclusion $\Delta \subseteq \Delta^\perp$. Then, a stabilizer code coming $E_\Delta^J$ can be constructed. Its parameters are $[[n_J, k, \geq d]]_q$, where
    $n_J = \prod_{j \not \in J} N_j \prod_{j \in J} (N_j-1)$, $k= n_J - 2 \; \mathrm{card} (\Delta)$ and $d= d\left((E_\Delta^J)^\perp \right)$.
    \item  Consider $s$ a positive integer that divides $r$ and subfield-subcodes with respect to the field $\mathbb{F}_{p^s}$.  Assume that $\mathfrak{I}_\mathbf{a} \cap \Delta^\perp \neq \emptyset$ whenever $\mathfrak{I}_\mathbf{a} \subseteq \Delta$. Then, a stabilizer code coming from $E_\Delta^{J,\sigma}$ can be constructed. Its parameters are $[[n_J, \geq k, \geq d]]_{p^s}$, where $n_J$ is as above, $k= 2 \sum_{\mathbf{a} \in \mathcal{A}| \mathfrak{I}_\mathbf{a} \cap \Delta^\perp \neq \emptyset} i_\mathbf{a} - n_J$ and $d= d\left((E_\Delta^{J,\sigma})^\perp \right)$.
    \item Let $s$ and $\Delta$ be as in {\rm (2)}. Suppose also that $\Delta \subseteq \mathcal{H}'$. Then, the parameters of the corresponding stabilizer code are $[[n_J, k, \geq d]]_{p^s}$, where $n_J = \prod_{j \not \in J} N_j \prod_{j \in J} (N_j-1)$, $k = n_J -2 \; \sum_{\mathfrak{I}_\mathbf{a} |\mathfrak{I}_\mathbf{a} \subseteq \Delta} i_\mathbf{a}$ and $d= d\left((E^{J,\sigma}_\Delta)^\perp \right)$.
\end{enumerate}
\end{teo}

Notice that the condition $\Delta \subseteq \Delta^\perp$ for sets $\Delta$ containing the element $\boldsymbol{0}$  can only  happen when $p \mid N_j$, for some $j \notin J$.

Later, in Section \ref{ejemplos}, we will provide some examples of quantum codes with good parameters.  Now, without any pretension on parameters and only for ease of reading, we give a simple example of stabilizer codes constructed with the tools of this section.
\begin{exa}
\label{ejem1}
{\rm With the above notation, consider $p=2, r=4, m=2, N_1=4, N_2=6$ and set $J=\{2\} \subseteq \{1,2\}$. It is clear that $R_J = \mathbb{F}_{q}[X_1, X_2] / \langle X_1^4 -X_1, X_2^5 - 1\rangle$, $T_1=3$, $T_2 = 4$, $\mathcal{H}_J =\{0,1,2,3\} \times \{0,1,2,3,4\}$ and $\mathcal{H'} =\{0,1,2\} \times \{0,1,2,3,4\}$.

If we consider the subset of $\mathcal{H}_J$
\[
\Delta := \{(0,1), (0,2),(0,3),(0,4),(1,2),(1,3),(2,0),(2,1),(2,4)\},
\]
then it is clear that $\Delta \subseteq \mathcal{H'}$ and by the the paragraph before Proposition \ref{pro:str}, $\Delta^\perp = \Delta \cup \{(0,0), (3,0)\}$. Then $\Delta \subseteq \Delta^\perp$, Items (1) in Proposition \ref{pro:str} and Theorem \ref{encero} and \cite{magma} determine a $[[20, 20 -2\cdot9,4]]_{16}$ code because the cardinality of $\Delta$ is nine.

With respect to subfield-subcodes, set $s=1$, then the minimal cyclotomic sets are: $\mathfrak{I}_{(0,0)} = \{(0,0)\}$,
$\mathfrak{I}_{(0,1)} = \{(0,1),(0,2),(0,3),(0,4)\}$,
$\mathfrak{I}_{(1,0)} = \{(1,0),(2,0)\}$,
$\mathfrak{I}_{(1,1)} = \{(1,1),(2,2),(1,4),(2,3)\}$,
$\mathfrak{I}_{(1,2)} = \{(1,2),(2,4),(1,3),(2,1)\}$,
$\mathfrak{I}_{(3,0)} = \{(3,0)\}$,
$\mathfrak{I}_{(3,1)} = \{(3,1),(3,2),(3,3),(3,4)\}$. Consider the set $\Delta_1 = \mathfrak{I}_{(0,1)} \cup \mathfrak{I}_{(1,2)}$, where $i_{(0,1)}=4, i_{(1,2)}=4$  and as, $(0,1) \in \mathfrak{I}_{(0,1)} \cap \Delta_1^\perp$ and $(1,2) \in \mathfrak{I}_{(1,2)} \cap \Delta_1^\perp$, the inclusion  $E_{\Delta_{1}}^{J,\sigma} \subseteq C_{\Delta_{1}}^{J,\sigma}$ holds by Item (2) of Theorem \ref{eraeldoce}. Finally, $\Delta_1 \subseteq \mathcal{H'}$ and Statement (3) in Theorem \ref{encero} shows that we can construct a $[[20, 20 -2\cdot(4+4),4]]_{2}$ stabilizer code.
}
\end{exa}

\section{Stabilizer $J$-affine variety codes: Hermitian  inner product}
We have just studied stabilizer codes determined  by $J$-affine variety codes which are self-orthogonal with respect to the Euclidean inner product. Next we describe what happens when one considers the Hermitian inner product.
\subsection{Hermitian duality for affine variety codes}
In this section our ring of polynomials  is $\mathbb{F}_{q^2}[X_1, X_2, \ldots, X_m]$ where, as above, $q= p^r$ and fix $m$ integers $N_j >1$, $1 \leq j \leq m$, such that each $N_j -1$ divides $q^2-1$. Following Section \ref{SS1}, we define the rings $R_J$ as quotients of the above ring. Now we state our first result. 

\begin{pro}
\label{dualH}
Let $J \subseteq \{ 1 , 2, \ldots , m\}$, consider $\boldsymbol{a}, \boldsymbol{b} \in \mathcal{H}_J$ and let $X^{\boldsymbol{a}}$ and $X^{\boldsymbol{b}}$ be two monomials representing elements in $R_J$. Then, the Hermitian inner product $\mathrm{ev}_J ( X^{\boldsymbol{a}}) \cdot_h \mathrm{ev}_J (X^{\boldsymbol{b}})$ is not $0$ if,  and only if, the following two conditions happen.
\begin{itemize}
\item For every $j \in J$, it holds that $q a_j + b_j \equiv 0 \mod (N_j -1)$, (i.e  $b_j = - q a_j + \lambda(N_j -1)$, for some $\lambda \ge 0$).
\item For every $j \notin J$, it holds that \begin{itemize}
\item either $a_j  + b_j > 0$ and $q a_j + b_j \equiv 0 \mod (N_j -1)$\\  (i.e.,  $b_j = - q a_j + \lambda(N_j -1)$, for some $\lambda > 0$, if $0 < a_j, b_j < N_j -1$, or $(a_j,b_j) \in \left\{(0,N_j -1), (N_j -1,0), (N_j -1,N_j -1)  \right\}$, otherwise);
\item or $a_j = b_j = 0$ and $p \not | ~ N_j$.
\end{itemize}
\end{itemize}
\end{pro}
\begin{proof}
It follows from Proposition \ref{dual} after taking into account that
$\mathrm{ev}_J ( X^{\boldsymbol{a}}) \cdot_h \mathrm{ev}_J (X^{\boldsymbol{b}})= \mathrm{ev}_J ( X^{\boldsymbol{a}}) \cdot \mathrm{ev}_J (X^{q \cdot \boldsymbol{b}})$, and  $a_j + q b_j \equiv 0 \mod (N_j -1)$ if, and only if, $b_j \equiv -q a_j \mod (N_j -1)$. Notice that this last equivalence happens because $N_j-1$ divides $q^2-1$ and thus $q$ is the inverse of $q$ modulo $N_j -1$.
\end{proof}


Consider a set $\Delta \subseteq \mathcal{H}_J$ and an element $\mathbf{a}$ in $\Delta$ as in Section \ref{laprimera}. Recall that, now, our field is $\mathbb{F}_{q^2}$. Suppose that $\Delta \subseteq \mathcal{H'}$, and, for each $j$, set $[-qa_j]_{N_j -1}$ a suitable representant of the congruence class modulo $N_j -1$ given by $-q a_j$. Then, we define $\Delta^{\perp_h}$ as the set
$$ \mathcal{H}_J \setminus \{ ([-qa_1]_{N_1 -1}, [-qa_2]_{N_2 -1}, \ldots, [-qa_m]_{N_m- 1})\; | \; \boldsymbol{a} \in \Delta\}.$$
Otherwise,  $\Delta^{\perp_h}$ is defined as
$$
\mathcal{H}_J \setminus \left\{ \{ ([-qa_1]_{N_1 -1}, [-qa_2]_{N_2 -1}, \ldots, [-qa_m]_{N_m- 1}) | \boldsymbol{a} \in \Delta \cap \mathcal{H'} \} \cup \{ \boldsymbol{a}' | \boldsymbol{a} \in \Delta , \boldsymbol{a} \notin \mathcal{H'} \} \right\},$$
where $\boldsymbol{a}'$ is a multi-valued vector defined by $a'_j = [-qa_j]_{N_j -1}$ if $a_j \notin \{ 0, N_j -1$\}, $a'_j$ is equal to $N_j -1$ if $a_j=0$ and $a'_j$ admits two values which are $N_j -1$ and $0$ if $a_j = N_j-1$.

Next we give a result about Hermitian duality of our codes which can be deduced from Proposition \ref{dualH}.

\begin{pro}
\label{laprocinco}
Let $\Delta \subseteq \mathcal{H}_J$ be as above.
\begin{enumerate}
\item Assume $\Delta \subseteq \mathcal{H'}$. Then the  equality of codes $(E^J_\Delta)^{\perp_h} = E^J_{\Delta^{\perp_h}}$ holds.
\item Otherwise, $\Delta \not \subseteq \mathcal{H'}$, it happens that $E^J_{\Delta^{\perp_h}} \subseteq (E^J_\Delta)^{\perp_h}$.
\end{enumerate}
\end{pro}

A necessary condition for the inclusion of a generalized Reed-Muller code over $\mathbb{F}_{q^2}$ into its Hermitian dual is given in \cite{Sarvepalli}. We conclude this section with the following result which proves that such a condition is also sufficient.

\begin{pro}
Set $RM_{q^2}(r,m)$ the $(r,m)$-generalized Reed-Muller code over the finite field $\mathbb{F}_{q^2}$. Then, the codes' inclusion $RM_{q^2}(r,m) \subseteq \left( RM_{q^2}(r,m) \right)^{\perp_h}$ holds if, and only if, $0 \leq r \leq m(q-1) -1$.
\end{pro}
\begin{proof}
By \cite{Sarvepalli}, it suffices to prove that $r > m(q-1) -1$ implies $$RM_{q^2}(r,m) \not \subseteq \left( RM_{q^2}(r,m) \right)^{\perp_h}.$$ Indeed, consider the $m$-tuple $\mathbf{q-1} = (q-1,q-1, \ldots, q-1)$  that provides the monomial $X^{\mathbf{q-1}}$. Clearly $\mathrm{ev}_\emptyset (X^{\mathbf{q-1}}) \in RM_{q^2}(r,m)$, however $\mathrm{ev}_\emptyset (X^{\mathbf{q-1}}) \not \in \left(RM_{q^2}(r,m)\right)^{\perp_h}$ because $q-1 + q(q-1) = q^2 -1$ which, by Proposition \ref{dualH}, proves that $\mathrm{ev}_\emptyset (X^{\mathbf{q-1}}) \cdot_h \mathrm{ev}_\emptyset (X^{\mathbf{q-1}}) \neq 0$. This concludes the proof.
\end{proof}

\subsection{Results on stabilizer codes using Hermitian inner product}
In this section we prove that, considering duality with respect to the inner Hermitian product, an analogous result to Theorem \ref{encero}   holds. As above $q=p^r$ and $s$ is a positive integer that divides $r$. The ground field of our evaluation codes is $\mathbb{F}_{q^2}$ and  we consider subfield-subcodes over $\mathbb{F}_{p^{2s}}$. Recall that $N_j -1$ divides $q^2 -1$ for all $j$. The trace maps  are defined as: $\mathrm{tr}_{2r}^{2s}: \mathbb{F}_{p^{2r}} \rightarrow \mathbb{F}_{p^{2s}}$,  $\mathrm{tr}_{2r}^{2s}(x)= x + x^{p^{2s}} + \cdots + x^{p^{2s(\frac{r}{s}-1)}}$; $\mathbf{tr}: \mathbb{F}_{p^{2r}}^{n_J} \rightarrow \mathbb{F}_{p^{2s}}^{n_J}$, determined by  $\mathrm{tr}_{2r}^{2s}$ componentwise and $\mathcal{T}: R_J \rightarrow R_J$, $\mathcal{T}(f) = f + f^{p^{2s}} + \cdots + f^{p^{2s(\frac{r}{s}-1)}}$.

Let us state our before mentioned result for codes constructed using the Hermitian inner product.

\begin{teo}
\label{esteo6}
Let $N_j$, $1 \leq j \leq m$, be positive integers such that $N_j -1 $ divides $p^{2r} -1$ for all index $j$. Let $\Delta$ be a subset of the above defined set $\mathcal{H}_J$.
\begin{enumerate}
\item Assume the set inclusion $\Delta \subseteq \Delta^{\perp_h}$. Then, a stabilizer code coming from $E_\Delta^J$ can be constructed and their parameters can be obtained with the same formulae given in Item (1) of Theorem \ref{encero} but replacing $\perp$ with $\perp_h$.
    \item  Consider a positive integer $s$  dividing $r$ and subfield-subcodes with respect to the field $\mathbb{F}_{p^{2s}}$.  Assume that $\mathfrak{I}_\mathbf{a} \cap \Delta^{\perp_h} \neq \emptyset$ whenever $\mathfrak{I}_\mathbf{a} \subseteq \Delta$. Then, a stabilizer code coming from $E_\Delta^{J,\sigma}$ can be constructed. Formulae in Item (2) (respectively (3), whenever $\Delta \subseteq \mathcal{H'}$) of Theorem \ref{encero}, replacing $\perp$ with $\perp_h$, determine the parameters of these codes.
\end{enumerate}
\end{teo}

\begin{proof}
Item (1) follows from Statement (2) of Theorem \ref{bueno} and similar arguments to those given in Section \ref{laprimera}.  With respect to Item (2), we prove our second statement since the unique difference with respect the first one relies in Proposition \ref{laprocinco} and we can only ensure $(E_\Delta^{J})^{\perp_h} = E_{\Delta^{\perp_h}}$ when $\Delta \subseteq \mathcal{H'}$.

Delsarte theorem is stated for Euclidean dual. Let us show that it is  true in our case and so our result is proved with the same arguments given in Section \ref{laprimera} and (2) of Theorem \ref{bueno}.
We start by proving the inclusion $(E_\Delta^{J,\sigma})^{\perp_h} \supseteq \mathbf{tr} (E_\Delta^{J})^{\perp_h} $. Indeed, to do it, it suffices take $\mathbf{a} \in (E_\Delta^{J})^{\perp_h}$ and $\mathbf{b} \in E_\Delta^{J,\sigma}$ and consider the following chain of equalities
$$\mathbf{tr}(\mathbf{a}) \cdot_h \mathbf{b}= \mathbf{tr}(\mathbf{a}) \cdot \mathbf{b}^q = \mathrm{tr}_{2r}^{2s}(\mathbf{a} \cdot \mathbf{b}^q) = \mathrm{tr}_{2r}^{2s}(\mathbf{a} \cdot_h \mathbf{b}) = \mathrm{tr}_{2r}^{2s} (0) = 0.$$

Finally, we prove that the dimensions of the vector spaces over $\mathbb{F}_{p^{2s}}$,  $(E_\Delta^{J,\sigma})^{\perp_h}$ and $\mathbf{tr} (E_\Delta^{J})^{\perp_h} $, coincide, which concludes the proof. Write $\Delta = \Delta_1 \cup \Delta_2$ where $\Delta_1$ is the union of the minimal cyclotomic sets $\mathfrak{I}_\mathbf{a}$ which are included in $\Delta$. $\Delta_2$ does not contain any complete set $\mathfrak{I}_\mathbf{a}$. Theorem \ref{eselteotres} proves that the dimension of the vector space $(E_\Delta^{J,\sigma})^{\perp_h}$ is $n_J - \mathrm{card} \; \Delta_1$. Now, consider the set $ \mathcal{H}_J \setminus \{ ([-qa_1]_{N_1 -1}, [-qa_2]_{N_2 -1}, \ldots, [-qa_m]_{N_m- 1})\; | \; \boldsymbol{a} \in \Delta\}$ and notice that the set of tuples $([-qb_1]_{N_1 -1}, [-qb_2]_{N_2 -1}, \ldots, [-qb_m]_{N_m- 1})$, defined by the elements $\boldsymbol{b}$ in a minimal cyclotomic set  $\mathfrak{I}_\mathbf{a}$, determine a minimal cyclotomic set of the same size, which we denote $\mathfrak{I}_{-q\mathbf{a}}$. Moreover, $\mathfrak{I}_\mathbf{a} \neq \mathfrak{I}_\mathbf{a'}$ implies $\mathfrak{I}_{-q\mathbf{a}} \neq \mathfrak{I}_{-q\mathbf{a'}}$.
  Taking into account that
 \[
 \dim \mathbf{tr} (E_\Delta^{J})^{\perp_h} = \sum_{\mathbf{a} \in \mathcal{A}| \mathfrak{I}_\mathbf{a} \cap \Delta^{\perp_h} \neq \emptyset} i_\mathbf{a}
 \]
 and that $\mathfrak{I}_\mathbf{a} \cap \Delta^{\perp_h} = \emptyset$ if, and only if, $\mathbf{a} = - q \mathbf{c}$, $\mathfrak{I}_\mathbf{c} \subset \Delta_1$, we deduce that $\dim \mathbf{tr} (E_\Delta^{J})^{\perp_h} = n_J - \mathrm{card} \; \Delta_1$ and our proof is finished.
\end{proof}

We remark that, as in the previous section, the condition $\Delta \subseteq \Delta^{\perp_h}$ for sets $\Delta$ containing the element $\boldsymbol{0}$ can only happen when $p \mid N_j$, for some $j \notin J$.

As above, we give an example only to facilitate the readability of this section. Examples with good parameters will be found in the next section.
\begin{exa}
{\rm
With the previous notations, set $p=2, r=4, s=2, N_1=4, N_2=6$. Also, $m=2$ and $J=\{2\}$. We deduce our example from the second statement in Theorem \ref{esteo6}. The minimal cyclotomic sets are $\{(0,0)\}$, $\{(0,1),(0,4)\}$, $\{(0,2),(0,3)\}$, $\{(1,0)\}$, $\{(1,4),(1,1)\}$, $\{(1,3),(1,2)\}$, $\{(2,0)\}$,
$\{(2,1),(2,4)\}$, $\{(2,2),(2,3)\}$, $\{(3,0)\}$, $\{(3,2),(3,3)\}$, $\{(3,4),(3,1)\}$.  The set $\Delta$ in Example \ref{ejem1} cannot be used now, since $(\mathfrak{I}_{(2,0)} = \{(2,0)\}) \cap \Delta^{\perp_h} = \emptyset$. To prove it, it suffices to recall that $s=2$ and to apply the paragraph after the proof of Proposition \ref{dualH}.

Finally, consider the three minimal cyclotomic sets $\mathfrak{I}_{(0,1)}= \{(0,1),(0,4)\}$, $\mathfrak{I}_{(0,2)}= \{(0,2),(0,3)\}$, $\mathfrak{I}_{(2,1)}= \{(2,1),(2,4)\}$ and the set $\Delta_2= \mathfrak{I}_{(0,1)} \cup  \mathfrak{I}_{(0,2)} \cup \mathfrak{I}_{(2,1)}$ which satisfies the requirements in Theorem \ref{esteo6} because, by the above mentioned paragraph, to determine the Hermitian dual, each element in $\Delta_2$ erases from $\mathcal{H}_J$ another one which is not in $\Delta_2$. For instance,  $(0,4)$ and $(2,4)$, erase $(3,2)$ and $(2,2)$, respectively. Each minimal cyclotomic set has two elements and therefore, by \cite{magma} and Item (2) in Theorem \ref{esteo6}, we get a  $[[20, 20 -2 (2+2+2),3]]_2$ code.
}
\end{exa}

We conclude this section with a short remark on decoding of our codes.
\begin{rem}
{\rm
Since classical methods of error correction can be adapted to decode quantum codes \cite{20kkk, Steane-2,71kkk}, we briefly comment on the decoding of affine variety codes. The literature contains some decoding procedures for affine variety codes \cite{FL, Mar}, a subclass of $J$-affine variety codes, which we believe that could be easily adapted to decode $J$-affine variety codes as well. More efficient decoding procedures, which correct up to the Feng-Rao bound, have been described for affine variety codes defined by order functions (see \cite{geil2} and references therein). It would be interesting to get self-orthogonal order domain codes providing good stabilizer codes, and investigate whether our examples are given by codes of this type.}
\end{rem}
\section{Some good quantum codes}
\label{ejemplos}

\subsection{Stabilizer codes  with Euclidean inner product}
We devote this section to give some examples of stabilizer codes obtained applying Theorems \ref{ham}, \ref{elnuevo} and \ref{encero}, with the help of \cite{magma}. We first provide parameters of some stabilizer codes over $\mathbb{F}_2$. These codes come from subfield-subcodes of $J$-affine variety codes. With the above notation, set $p=2$, $r=7$, $s=1$, $N_1=128$ and consider codes $\mathcal{C}_i = C^\sigma_{\Delta_i}$, $\hat{\mathcal{C}}_i = C^\sigma_{\hat{\Delta}_i}$, $i=1,2$ and $\mathcal{C}_3 = C^\sigma_{\Delta_3}$, where we have omitted the super-index $J=\{1\}$. Table \ref{tabla1} shows their defining sets $\Delta$ and  parameters (as linear codes).

\begin{table}[ht]
\centering
\begin{tabular}{||c|c|c|c|c||}
  \hline \hline
  Code  & $n$ & $k$ & $d$  & Defining set $\Delta$ \\
  \hline \hline
  $\mathcal{C}_1$  & 127 & 85 & 12 &  $\Delta_1 = \{42, 84, 41, 82, 37, 74, 21,
2, 4, 8, 16, 32, 64, 1, 6,$\\
     & &  &  &
   $ 12, 24, 48, 96, 65, 3, 10, 20, 40, 80, 33, 66, 5, 14, 28,$\\
     & &  &  &
   $56, 112, 97, 67, 7, 18, 36, 72, 17, 34, 68, 9\}$  \\
  \hline
    $\hat{\mathcal{C}}_1$ & 127 & 91 &  12 &  $\hat{\Delta}_1 =\{0,
2, 4, 8, 16, 32, 64, 1 ,6, 12, 24, 48, 96, $\\
     & &  &  &
   $  65, 3, 10, 20, 40, 80, 33, 66, 5, 14, 28,$\\
     & &  &  &
   $56, 112, 97, 67, 7, 18, 36, 72, 17, 34, 68, 9\}$  \\
 \hline
    $\mathcal{C}_2$ & 127 & 99 & 8 &  $\Delta_2 =\{42, 84, 41, 82, 37, 74, 21,
2, 4, 8, 16, 32, 64, 1, 6,$\\
     & &  &  &
   $ 12, 24, 48, 96, 65, 3, 10, 20, 40, 80, 33, 66, 5\}$\\
     \hline
   $\hat{\mathcal{C}}_2$ & 127 & 105 & 8 &  $\hat{\Delta}_2 =\{0,
2, 4, 8, 16, 32, 64, 1 ,6, 12, 24, 48, 96, $\\
     & &  &  &
   $  65, 3, 10, 20, 40, 80, 33, 66, 5\}$\\
\hline
$\mathcal{C}_3$ & 127 & 106 & 7 &  $\Delta_3 =\{
2, 4, 8, 16, 32, 64, 1, 6, 12, 24, 48, 96, $\\
     & &  &  &
   $  65, 3, 10, 20, 40, 80, 33, 66, 5\}$\\
\hline
 \hline
\end{tabular}
\caption{$J$-affine variety codes over $\mathbb{F}_2$}
\label{tabla1}
\end{table}
Theorem \ref{elnuevo} applied to these codes provides the stabilizer code $C_1$. Table \ref{ta:best} displays the parameters of this stabilizer quantum code and several expurgations. According to \cite{codet}, the parameters of the codes in Table \ref{ta:best} improve the parameters of the best known binary quantum codes, and thus they are records.
\begin{table}[ht]
\centering
\begin{tabular}{||c|c|c|c|c||}
  \hline \hline
  Code  & $n$ & $k$ & $d \geq$ & Distance in \cite{codet}  \\
  \hline \hline
  $C_1$ & 127 & 63 & 12  & 11  \\
  \hline
   $C_2$ = Subcode$(C_1,62)$ & 127 & 62 & 12 & 11  \\
 \hline
  $C_3$ = Subcode$(C_1,61)$ & 127 & 61 & 12 & 11  \\
 \hline
  $C_{4}$ = Subcode$(C_1,60)$ & 127 & 60 & 12 & 11  \\
 \hline
 $C_{5}$ = Subcode$(C_1,59)$ & 127 & 59 & 12 & 11  \\
 \hline
  \hline
\end{tabular}
\caption{New records of quantum codes over $\mathbb{F}_2$}
\label{ta:best}
\end{table}

Consider now $\mathbb{F}_3$ as a ground field. Table \ref{tabla4} shows defining sets, values $p, r, s, N_j$ and sets $J$ as above defined to determine stabilizer codes coming from subfield-subcodes of $J$-affine variety codes.  They are obtained following Item (3) in Theorem \ref{encero}. The corresponding parameters are displayed in Table \ref{tabla3}. Parameters of the codes given by Steane enlargement, SE, can be seen in Table \ref{tabla5}. Notice that all these codes exceed the different known versions of the (quantum) Gilbert-Varshamov bound \cite{eck,mat,feng}, \cite[Lemma 31]{kkk}, which is noted in the tables by saying that are of type GV.
\begin{table}[htb]
\centering
\begin{tabular}{||c|c|c|c||c|c|c|c||}
  \hline \hline
  Code / Subset & $n$ & $k$ & $d \geq $ & Code / Subset & $n$ & $k$ & $d \geq $ \\
  \hline \hline
  $C_1$ / $\Delta_{1}$ & 144 & 132 & 3  &   $C_2$ / $\Delta_{2}$ & 144 & 126 & 4  \\
  \hline
  $C_3$ / $\Delta_{3}$ & 72 & 60 & 2  &  $C_4$ / $\Delta_{4}$ & 72 & 62 & 3  \\
  \hline
   $C_5$ / $\Delta_{5}$ & 72& 56 & 4  & $C_6$ / $\Delta_{6}$ & 72 & 44 & 6 \\
  \hline
  \hline
\end{tabular}
\captionof{table}{Stabilizer codes over $\mathbb{F}_3$}
\label{tabla3}
\end{table}
\begin{table}[htb]
\centering
\begin{tabular}{||c|c|c|c|c|c|c|c||}
  \hline \hline
 Subset & $p$ & $r$ & $s$ & $N_1-1$  & $N_2-1$ & $N_3-1$ & Set $J$\\
  \hline \hline
  $
  \begin{array}{r}
  \Delta_1= \{ (0, 0, 0 ), (7, 6, 1 ), (5, 2, 1), \\ (0, 3, 1), (0, 1, 1), (0, 4, 1)
\}\\
   \end{array}
   $
    & 3 &2 & 1 & 8 & 8& 2 & $\{2,3\}$ \\
  \hline
   $
  \begin{array}{r}
  \Delta_2= \{ (0, 0, 0 ), (7, 6, 1 ), (5, 2, 1), \\ (0, 3, 1), (0, 1, 1), (0, 4, 1), \\  (0, 0, 1 ),
  (6, 3, 0 ), (2, 1, 0)
 \}\\
   \end{array}
   $
    & 3 &2 & 1 & 8 & 8 & 2 & $\{2,3\}$\\
  \hline
     $
  \begin{array}{r}
  \Delta_3= \{ (0,4 )\}\\
   \end{array}
   $
  & 3 &4 & 1 & 8 & 8 & - & \{2\}  \\
  \hline
       $
  \begin{array}{r}
  \Delta_4= \{ ( 0,4 ), (0, 7 ),( 0, 5),( 7, 4 ),( 5, 4)
 \}\\
   \end{array}
   $
     & 3 &4 & 1 & 8 & 8 &  - & \{2\}  \\
  \hline
       $
  \begin{array}{r}
  \Delta_5= \{ ( 0,4 ), (0, 7 ),( 0, 5),( 7, 4 ), \\ ( 5, 4),  (0, 0),(4, 7),( 4, 5)
 \}\\
   \end{array}
   $
 & 3 &4 & 1 & 8 & 8 &  - & \{2\} \\
  \hline
   $
  \begin{array}{r}
  \Delta_6= \{  ( 0,4 ), (0, 7 ),( 0, 5),( 7, 4 ),( 5, 4),  \\ (0, 0),(4, 7),( 4, 5),(3, 7 ),(1, 5), \\ (0, 6),( 0, 2),(6, 5),(2, 7 )
 \}\\
   \end{array}
   $
     & 3 &4 & 1 & 8 & 8 &  - & \{2\}\\
  \hline
  \hline
   \end{tabular}
   \caption{Defining sets of $J$-affine variety codes over $\mathbb{F}_3$}
\label{tabla4}
\end{table}
\begin{table}[htb]
\centering
\begin{tabular}{||c|c|c|c|c||}
  \hline \hline
  Code  & $n$ & $k$ & $d \geq$   & Type\\
  \hline \hline
  $C_7$ =  $\mathrm{SE} (C_2,C_1)$& 144 & 129 & 4  & GV  \\
   \hline
 $C_{8}$ =  $\mathrm{SE} (C_4,C_3)$& 72 & 66 & 3 & GV   \\
  \hline
 $C_{9}$ =  $\mathrm{SE} (C_5,C_4)$& 72 & 59 & 4& GV   \\
  \hline
 $C_{10}$ =  $\mathrm{SE} (C_6,C_5)$& 72 & 50 & 6 & GV   \\
 \hline
  \hline
\end{tabular}
\caption{Stabilizer codes over $\mathbb{F}_3$ exceeding the Gilbert-Varshamov bounds. Obtained from codes $C_i$, $1 \leq i \leq 6$, in Table 3}
\label{tabla5}
\end{table}

Finally we use Theorem \ref{elnuevo} to give a stabilizer code $C$ over the field $\mathbb{F}_4$ with parameters $[[63,45, \geq 6]]_4$, which is of type GV. Notice that La Guardia in \cite{lag3} (see also \cite{ham}) gives two stabilizer codes with parameters $[[63,42, \geq 6]]_4$  and $[[63,46, \geq 5]]_4$. Our code improves the first one and has relative parameters better than the second one. To construct $C$, it suffices to take values $p=2$, $r=6$, $s=2$ and $N_1=64$ and apply Theorem \ref{elnuevo} with respect to the affine variety codes $\mathcal{C}_i = C^\sigma_{\Delta_i}$, $\hat{\mathcal{C}}_i = C^\sigma_{\hat{\Delta}_i}$, $i=1,2$ and $\mathcal{C}_3 = C^\sigma_{\Delta_3}$, again the super-index $J=\{1\}$ is omitted. Table \ref{tablanueva} shows the sets $\Delta$ and their parameters. Notice that codes and parameters in Table \ref{tablanueva} correspond to linear codes.

\begin{table}[ht]
\centering
\begin{tabular}{||c|c|c|c|c||}
  \hline \hline
  Code  & $n$ & $k$ & $d$  & Defining set $\Delta$ \\
  \hline \hline
  $\mathcal{C}_1$  & 63 & 52 & 6 &  $\Delta_1= \{0, 21, 8, 32, 2, 40, 34, 10, 62, 59, 47\}$\\
  \hline
    $\hat{\mathcal{C}}_1$ & 63 & 53 &  6 &  $\hat{\Delta}_1 = \{ 21, 8, 32, 2, 40, 34, 10, 62, 59, 47\}$\\
 \hline
    $\mathcal{C}_2$ & 63 & 55 & 5 &  $\Delta_2 = \{ 0, 21, 8, 32, 2, 40, 34, 10 \}$\\
     \hline
   $\hat{\mathcal{C}}_2$ & 63 & 56 &  4 &  $\hat{\Delta}_2 =  \{ 21, 8, 32, 2, 40, 34, 10\}$\\
\hline
$\mathcal{C}_3$ & 63 & 56 & 4 &  $\Delta_3 = \{21, 8, 32, 2, 40, 34, 10\}$\\
\hline
 \hline
\end{tabular}
\caption{$J$-affine variety codes over $\mathbb{F}_4$  that produce a $[[63,45, \geq 6]]_4$ quantum code  by Theorem \ref{elnuevo}}
\label{tablanueva}
\end{table}

\subsection{Stabilizer codes  with the Hermitian inner product}

This section gives examples of stabilizer codes obtained following Theorem \ref{esteo6}. They are constructed from subfield-subcodes of $J$-affine variety codes and we have considered duality with respect to the Hermitian inner product. We group them in tables corresponding to the same ground field. We display first the parameters and the type (GV or not) and afterwards the defining set $ \Delta$ and the corresponding values $p, r, s, N_j$ and sets $J$. Tables \ref{tabla12} and \ref{tabla13} (respectively \ref{tabla10} and \ref{tabla11}, \ref{tabla14} and \ref{tabla15}, \ref{tabla6} and \ref{tabla7}, \ref{tabla8} and \ref{tabla9}) correspond to stabilizer codes over $\mathbb{F}_2$ (respectively, $\mathbb{F}_3$, $\mathbb{F}_4$, $\mathbb{F}_5$, $\mathbb{F}_7$).

We conclude by adding that the codes in Section \ref{ejemplos} improve  the parameters of those codes in \cite{edel} which have the same length. In addition, our code in Table \ref{tabla8} with parameters $[[144,134, \geq 4]]_7$ also improves the parameters $[[144,132, \geq 4]]_7$ which can be obtained by applying \cite[Theorem 39]{lag2}.

\begin{table}[ht]
\centering
\begin{tabular}{||c|c|c|c|c||c|c|c|c|c||}
  \hline \hline
  Code / Subset & $n$ & $k$ & $d \geq $ & Type & Code / Subset & $n$ & $k$ & $d \geq $ & Type  \\
  \hline \hline
  $C_1$ / $\Delta_{1}$ & 225 & 205 & 4 & GV &   $C_2$ / $\Delta_{2}$ & 225 & 197 & 5 & GV \\
  \hline
  $C_3$ / $\Delta_{3}$ & 240 & 222 & 4 & GV &    &  &  &  &  \\
  \hline
  \hline
\end{tabular}
\captionof{table}{Stabilizer codes over $\mathbb{F}_2$}
\label{tabla12}
\end{table}
\begin{table}[ht]
\centering
\begin{tabular}{||c|c|c|c|c|c|c||}
  \hline \hline
 Subset & $p$ & $r$ & $s$ & $N_1-1$  & $N_2-1$ &  Set $J$\\
  \hline \hline
  $
  \begin{array}{r}
  \Delta_1= \{( 12, 5 ),( 3, 5 ),( 9, 13),( 6, 7 ),\\ (13, 13 ),( 7, 7 ),(5, 9 ),(5, 6),\\( 9, 0 ),( 6, 0 )
 \}\\
   \end{array}
   $
    & 2 &4 & 2 & 15 & 15 & \{1,2\}\\
  \hline
   $
  \begin{array}{r}
  \Delta_2= \{  ( 12, 5 ),(3, 5),(9, 13 ),( 6, 7 ),\\(13, 13 ),(7, 7 ),(5, 9 ),(5, 6 ),\\( 9, 0 ),(6, 0 ),( 10, 8),( 10, 2),\\(12, 12 ),(3, 3 )
 \}\\
   \end{array}
   $
    & 2 &4 & 2 & 15 & 15 & \{1,2\}\\
  \hline
    $
  \begin{array}{r}
  \Delta_3= \{ (4, 4 ),(1, 1 ),(0, 9 ),(0, 6),( 0, 14 ),\\(0, 11 ),( 8, 4),( 2, 1 ),( 0, 10 )
 \}\\
   \end{array}
   $
    & 2 &4 & 2 & 15 & 15& \{2\}\\
  \hline
  \hline
   \end{tabular}
   \captionof{table}{Defining sets of the  codes over $\mathbb{F}_2$ in Table \ref{tabla12}}
\label{tabla13}
\end{table}

\begin{table}[ht]
\centering
\begin{tabular}{||c|c|c|c|c||c|c|c|c|c||}
  \hline \hline
  Code / Subset & $n$ & $k$ & $d \geq $ & Type & Code / Subset & $n$ & $k$ & $d \geq $ & Type  \\
  \hline \hline
  $C_1$ / $\Delta_{1}$ & 40 & 32 & 4 & GV &   $C_2$ / $\Delta_{2}$ & 40 & 26 & 6 & GV \\
  \hline
  $C_3$ / $\Delta_{3}$ & 40 & 20 & 7 & GV &   $C_4$ / $\Delta_{4}$ & 40 & 16 & 8 & GV \\
    \hline
  $C_5$ / $\Delta_{5}$ & 45 & 33 & 4 & GV &   $C_6$ / $\Delta_{6}$ & 45 & 27 & 5 &  \\
  \hline
  $C_7$ / $\Delta_{7}$ & 81 & 73 & 4 & GV &   $C_8$ / $\Delta_{8}$ & 91 & 73 & 6 & GV \\
  \hline
  \hline
\end{tabular}
\captionof{table}{Stabilizer codes over $\mathbb{F}_3$}
\label{tabla10}
\end{table}
\begin{table}[ht]
\centering
\begin{tabular}{||c|c|c|c|c|c|c||}
  \hline \hline
 Subset & $p$ & $r$ & $s$ & $N_1-1$  & $N_2-1$ &  Set $J$\\
  \hline \hline
  $
  \begin{array}{r}
  \Delta_1= \{2,5,15,18 \}\\
   \end{array}
   $
    & 3 &4 & 2 & 40 & - & \{1\}\\
  \hline
   $
  \begin{array}{r}
  \Delta_2= \{  19, 11 , 15 ,36, 4 , 5 , 38, 22
 \}\\
   \end{array}
   $
     & 3 &4 & 2 & 40 & - & \{1\}\\
  \hline
    $
  \begin{array}{r}
  \Delta_3= \{  35, 18, 2 , 36, 4 , 14, 6 , 23, 7 , 25  \}\\
   \end{array}
   $
   & 3 &4 & 2 & 40 & - & \{1\}\\
  \hline
    $
  \begin{array}{r}
  \Delta_4= \{  18, 2 ,15 , 27, 3,24, 16 , 36, 4 ,
     5 , 14, 6  \}\\
   \end{array}
   $
   & 3 &4 & 2 & 40 & - & \{1\}\\

  \hline
   $
  \begin{array}{r}
  \Delta_5= \{
    ( 0, 0 ),(0, 3 ),( 0, 2 ),( 6, 4), \\( 6, 1),( 3, 0)\}\\
   \end{array}
   $
   & 3 &4 & 2 & 8 & 5 & \{2\}\\
  \hline
   $
  \begin{array}{r}
  \Delta_6= \{  (0, 4 ),(0, 1),( 0, 3),(0, 2),(1, 4), \\
  ( 1, 1),( 2, 4 ),( 2, 1 ),(3, 0)  \}\\
   \end{array}
   $
    & 3 &4 & 2 & 8 & 5 & \{2\}\\
    \hline
   $
  \begin{array}{r}
  \Delta_7= \{ 0 ,70,71,9 \}\\
   \end{array}
   $
    & 3 &4 & 2 & 80 & - & $\emptyset$ \\
    \hline
   $
  \begin{array}{r}
  \Delta_8= \{   9, 81, 1 ,
     50, 86, 46 ,
     54, 31, 6  \}\\
   \end{array}
   $
    & 3 &6 & 2 & 91 & - & \{1\}\\
  \hline
  \hline
   \end{tabular}
   \captionof{table}{Defining sets of the codes over $\mathbb{F}_3$ in Table \ref{tabla10}}
\label{tabla11}
\end{table}

\begin{table}[ht]
\begin{tabular}{||c|c|c|c|c||c|c|c|c|c||}
  \hline \hline
  Code / Subset & $n$ & $k$ & $d \geq $ & Type & Code / Subset & $n$ & $k$ & $d \geq $ & Type  \\
  \hline \hline
  $C_1$ / $\Delta_{1}$ & 51 & 41 & 4 & GV &   $C_2$ / $\Delta_{2}$ & 51 & 39 & 5 & GV \\
  \hline
  $C_3$ / $\Delta_{3}$ & 51 & 37 & 6 & GV &   $C_4$ / $\Delta_{4}$ & 51 & 36 & 7 & GV \\
 \hline
  $C_5$ / $\Delta_{5}$ & 52 & 44 & 4 & GV &   $C_6$ / $\Delta_{6}$ & 52 & 38 & 5 & GV \\
  \hline
  $C_7$ / $\Delta_{7}$ & 52 & 36 & 6 & GV &   $C_8$ / $\Delta_{8}$ & 255 & 245 & 4 & GV \\
   \hline
  $C_9$ / $\Delta_{9}$ & 54 & 44 & 4 & GV &   $C_{10}$ / $\Delta_{10}$ & 54 & 36 & 6 & GV \\
  \hline
  \hline
\end{tabular}
\captionof{table}{Stabilizer codes over $\mathbb{F}_4$ exceeding the Gilbert-Varshamov bounds}
\label{tabla14}
\end{table}
\begin{table}[ht]
\centering
\begin{tabular}{||c|c|c|c|c|c|c||}
  \hline \hline
 Subset & $p$ & $r$ & $s$ & $N_1-1$  & $N_2-1$ & Set $J$\\
  \hline \hline
  $
  \begin{array}{r}
  \Delta_1= \{34 ,32, 2 ,42, 9  \}\\
   \end{array}
   $
    & 2 &8& 4 & 51 & - & \{1\}\\
  \hline
   $
  \begin{array}{r}
  \Delta_2= \{     32, 2 ,45, 6 ,10, 7 \}\\
   \end{array}
   $
      & 2 &8& 4 & 51 & - & \{1\}\\
  \hline
    $
  \begin{array}{r}
  \Delta_3= \{    34 ,29, 5 ,26, 8 ,27, 24  \}\\
   \end{array}
   $
    & 2 &8& 4 & 51 & - & \{1\}\\
    \hline
    $
  \begin{array}{r}
  \Delta_4= \{  23, 11 ,39, 12 , 16, 1 , 34 , 50, 35
 \}\\
   \end{array}
   $
    & 2 &8& 4 & 51 & - & \{1\}\\
    \hline
  $
  \begin{array}{r}
  \Delta_5= \{ 0 ,34 , 26, 8  \}\\
   \end{array}
   $
    & 2 &8& 4 & 51 & - & $\emptyset$\\
  \hline
   $
  \begin{array}{r}
  \Delta_6= \{  0 ,32, 2 ,49, 19 , 30, 21  \}\\
   \end{array}
   $
      & 2 &8& 4 & 51& - & $\emptyset$\\
  \hline
    $
  \begin{array}{r}
  \Delta_7= \{    40, 28 ,0 , 34 , 48, 3 , 26, 8   \}\\
   \end{array}
   $
    & 2 &8& 4 & 51 & - & $\emptyset$\\
    \hline
    $
  \begin{array}{r}
  \Delta_8= \{     114, 39 , 17 ,241, 31  \}\\
   \end{array}
   $
    & 2 &8& 4 & 255 & - & \{1\} \\
    \hline
    $
  \begin{array}{r}
  \Delta_9= \{  ( 0, 0),(0, 1 ),(0, 2 ),( 16, 0),(1, 0)\}\\
   \end{array}
   $
    & 2 &8& 4 & 17 & 3 & \{2\}\\
    \hline
    $
  \begin{array}{r}
  \Delta_{10}= \{ (0, 0),( 0, 1 ),(13, 0 ),\\( 4, 0),(0, 2),( 12, 2),( 5, 2 ),\\(15, 1 ),( 2, 1 )  \}\\
   \end{array}
   $
  & 2 &8& 4 & 17 & 3 & \{2\}\\
  \hline
  \hline
   \end{tabular}
   \captionof{table}{Defining sets of the codes over $\mathbb{F}_4$ in Table \ref{tabla14}}
\label{tabla15}
\end{table}
\begin{table}[htb]
\centering
\begin{tabular}{||c|c|c|c|c||c|c|c|c|c||}
  \hline \hline
  Code / Subset & $n$ & $k$ & $d \geq $ & Type & Code / Subset & $n$ & $k$ & $d \geq $ & Type  \\
  \hline \hline
  $C_1$ / $\Delta_{1}$ & 52 & 36 & 6 & GV &   $C_2$ / $\Delta_{2}$ & 104 & 96 & 4 & GV \\
  \hline
  $C_3$ / $\Delta_{3}$ & 112 & 102 & 4 & GV &   $C_4$ / $\Delta_{4}$ & 156 & 148 & 4 & GV \\
  \hline
    $C_5$ / $\Delta_{5}$ & 72 & 62 & 4 & GV &   $C_6$ / $\Delta_{6}$ & 96 & 86 & 4 & GV \\
  \hline
  \hline
\end{tabular}
\captionof{table}{Stabilizer codes  over $\mathbb{F}_5$ exceeding the Gilbert-Varshamov bounds}
\label{tabla6}
\end{table}
\begin{table}[ht]
\centering
\begin{tabular}{||c|c|c|c|c|c|c|c||}
  \hline \hline
 Subset & $p$ & $r$ & $s$ & $N_1-1$  & $N_2-1$ & $N_3-1$ & Set $J$ \\
  \hline \hline
  $
  \begin{array}{r}
  \Delta_1= \{ 15, 11,32, 20 ,30, 22 ,17, 9
 \} \\
   \end{array}
   $
    & 5 &4 & 2 & 52 & - &- & \{1\} \\
  \hline
   $
  \begin{array}{r}
  \Delta_2= \{ (3, 0 ),(5, 0 ),( 0, 8),(0,5)
 \}\\
   \end{array}
   $
    & 5 &4 & 2 & 8 & 13 &- & \{1,2\}\\
  \hline
    $
  \begin{array}{r}
  \Delta_3= \{ ( 0,1 ),( 0,2 ),( 0, 7),(12,1),(1,1)
 \}\\
   \end{array}
   $
    & 5 &4 & 2 & 13 & 8 &- & \{2\}\\
  \hline
     $
  \begin{array}{r}
  \Delta_4= \{ ( 1,0 ),( 3,0 ),( 2,6),(2,7)
 \}\\
   \end{array}
   $
    & 5 &4 & 2 & 12 & 13 &- & \{1,2\}\\
  \hline
       $
  \begin{array}{r}
  \Delta_5= \{(2, 6, 2),( 1, 7, 2),(0, 5, 1),\\ ( 2, 2, 0),( 0, 3, 2)
 \}\\
   \end{array}
   $
    & 5 &4 & 2 & 3 & 8 &3 & \{1,2,3\} \\
  \hline
       $
  \begin{array}{r}
  \Delta_6= \{ ( 0, 7, 2 ),( 0, 2, 0),( 2, 5, 0), \\ ( 0, 5, 1),( 1, 2, 0)
 \}\\
   \end{array}
   $
    & 5 &4 & 2 & 3 & 8 &3 & \{2,3\}\\
  \hline
  \hline
   \end{tabular}
   \captionof{table}{Defining sets of the codes over $\mathbb{F}_5$ in Table \ref{tabla6}}
\label{tabla7}
\end{table}

\begin{table}[ht]
\centering
\begin{tabular}{||c|c|c|c|c||c|c|c|c|c||}
  \hline \hline
  Code / Subset & $n$ & $k$ & $d \geq $ & Type & Code / Subset & $n$ & $k$ & $d \geq $ & Type  \\
  \hline \hline
  $C_1$ / $\Delta_{1}$ & 90 & 78 & 4 & GV &   $C_2$ / $\Delta_{2}$ & 80 & 72 & 4 & GV \\
  \hline
  $C_3$ / $\Delta_{3}$ & 144 & 134 & 4 & GV &    &  &  &  &  \\

  \hline
  \hline
\end{tabular}
\captionof{table}{Stabilizer codes over $\mathbb{F}_7$ exceeding the Gilbert-Varshamov bounds}
\label{tabla8}
\end{table}
\begin{table}[ht]
\centering
\begin{tabular}{||c|c|c|c|c|c|c||}
  \hline \hline
 Subset & $p$ & $r$ & $s$ & $N_1-1$  & $N_2-1$  & Set $J$ \\
  \hline \hline
  $
  \begin{array}{r}
  \Delta_1= \{( 4, 0 ),( 5, 5),( 5, 9),( 5, 6),\\( 1, 5), ( 0, 10)
 \}\\
   \end{array}
   $
    & 7 &4 & 2 & 6 & 15 & \{1,2\}\\
  \hline
   $
  \begin{array}{r}
  \Delta_2= \{  (1, 0 ),(2, 4 ),( 2, 1),( 3, 0)
 \}\\
   \end{array}
   $
    & 7 &4 & 2 & 16 & 5 & \{1,2\}\\
  \hline
    $
  \begin{array}{r}
  \Delta_3= \{ (10, 2), (15, 2), (19, 1), (8, 0), (2, 2)\}\\
   \end{array}
   $
    & 7 &2 & 2 & 48 & 3 & \{1,2\}\\
  \hline
  \hline
   \end{tabular}
   \captionof{table}{Defining sets of the codes over $\mathbb{F}_7$ in Table \ref{tabla8}}
\label{tabla9}
\end{table}

\section*{Appendix}
We devote this appendix to prove Theorem \ref{elnuevo} which was stated in the introduction and preliminaries of this paper. To do this, we adapt to our purposes some facts described in \cite {Steane-E} and \cite{ham}. Consider the vector space $\mathbb{F}_q^{2n}$ and the symplectic inner product $(\mathbf{u}|\mathbf{v}) \cdot_s (\mathbf{u}'|\mathbf{v}') = \mathbf{u} \cdot \mathbf{v}' - \mathbf{v} \cdot \mathbf{u}'$. Recall that the weight $\mathrm{w}(\mathbf{u}|\mathbf{v})$  of a word  $(\mathbf{u}|\mathbf{v})$ as above is the number of indexes $i$, $1 \leq i \leq n$, such that either $u_i$ or $v_i$ (or both) are not zero, where the $u_i$ (respectively, $v_i$) represent the coordinates of the vector $\mathbf{u}$ (respectively, $\mathbf{v}$). Following \cite{AK} (see also \cite{19kkk}), to get our stabilizer code we only need to find a vector subspace $S$ in $\mathbb{F}_q^{2n}$ such that $S^{\perp_s} \subseteq S$ with dimension $k_2 + \hat{k}_1$ and minimum distance larger than or equal to that stated in the statement. Let us describe it. Set $G_1$ ($\hat{G}_1$, $L$, respectively) generator matrices of the codes $C_1$ ($\hat{C}_1$, $D$, respectively) and let $S$ be the code  of $\mathbb{F}_q^{2n}$ generated by the matrix
\[
\left(
    \begin{array}{cc}
      L & AL \\
      G_1 & 0 \\
      0 & \hat{G}_1 \\
    \end{array}
  \right),
  \]
where $A$ is a  fixed point free square matrix (see \cite {Steane-E} and \cite{ham} for its existence). Our hypotheses imply $\hat{k}_1 + k_2 = k_1 + \hat{k}_2$ and that the rows of the previous matrix are linearly independent, therefore, for computing the dimension of $S$, it suffices to see that the number of rows is $k_2 -k_1 + k_1 + \hat{k}_1 = k_2 + \hat{k}_1$.

Let $H_2$ ($\hat{H}_2$, respectively) be a parity check matrix of the code $C_2$ ($\hat{C}_2$, respectively), one can found a matrix $B$ such that
$$\left(
\begin{array}{c}
  H_2 \\
  B \\
   \end{array}
   \right),  \left( \left(
\begin{array}{c}
  \hat{H}_2 \\
  B \\
   \end{array}
   \right), \mathrm{respectively}
   \right)$$
   is a parity check matrix for $C_1$ (respectively, for $\hat{C}_1$). Now defining the matrix $K= BL^t (A^t)^{-1} (B L^t)^{-1}$,  it is not difficult to prove that
   \[
\left(
    \begin{array}{cc}
      KB & B \\
      \hat{H}_2 & 0 \\
      0 & H_2 \\
    \end{array}
  \right),
  \]
is a parity check matrix for the code $S$ and therefore one has that $S^{\perp_s} \subseteq S$.

To end our proof,  it only remains to study what happens with the weight $\mathrm{w}(\mathbf{u}|\mathbf{v})$ for elements $(\mathbf{u}|\mathbf{v}) \in S$. First assume $q=2$, a generic element in $S$ has the form $(\mathbf{v}_1 L +\mathbf{v}_2 G_1 | \mathbf{v}_1 AL +\mathbf{v}_3 \hat{G}_1)$, where $\mathbf{v}_1, \mathbf{v}_2, \mathbf{v}_3$ are suitable vectors with coordinates  in $\mathbb{F}_q$. When $\mathbf{v}_1$ is the zero vector, $\mathbf{u}$ must be in $C_1$ and $\mathbf{v}$ in $\hat{C}_1$, which proves that, in this case, $\mathrm{w}(\mathbf{u}|\mathbf{v})$ must be larger than or equal to the minimum of the values $d_1$ and $\hat{d}_1$. Otherwise, $\mathbf{v}_1 \neq \mathbf{0}$,  one can use the property
\[
\mathrm{w}(\mathbf{u}|\mathbf{v}) = \frac{\mathrm{wt}(\mathbf{u})+ \mathrm{wt}(\mathbf{v}) + \mathrm{wt}(\mathbf{u}+\mathbf{v})}{2},
\]
where $\mathrm{wt}$ denotes the standard weight of a word in a code in $\mathbb{F}_q^n$, and this concludes the proof since $\mathbf{u} \in C_2$, $\mathbf{v} \in \hat{C}_2$, $\mathbf{u}+\mathbf{v} \in C_3$ and the fact that $(C_1 + \hat{C}_1) \cap D = \{\mathbf{0}\}$ implies that none of the previous vectors are zero.

Let us consider $q \neq 2$, we will only study $\mathrm{w}(\mathbf{u}|\mathbf{v})$ for $\mathbf{v}_1 \neq 0$. For convenience assume that the coordinates $u_{t+1}, u_{t+2}, \ldots, u_n$ of the word $\mathbf{u}$ are zero and that this does not happen with the remaining coordinates. As showed in \cite{ham}, there exists $\lambda \in \mathbb{F}_q$ such that
\[
\mathrm{w}(\mathbf{u}|\mathbf{v}) = t + \mathrm{wt} (v_{t+1}, v_{t+2}, \ldots, v_n ) \geq \mathrm{wt}(\mathbf{v} - \lambda \mathbf{u}) + \frac{\mathrm{wt}(\mathbf{u})}{q}
\]
and, symmetrically, $\mathrm{w}(\mathbf{u}|\mathbf{v}) \geq \mathrm{wt}(\mathbf{u} - \lambda' \mathbf{v}) + \frac{\mathrm{wt}(\mathbf{v})}{q}$, for some $\lambda' \in \mathbb{F}_q$, holds. This finishes the proof because, as before, our hypotheses imply that $\mathbf{0} \neq \mathbf{v} - \lambda \mathbf{u}$ and $\mathbf{0} \neq \mathbf{u} - \lambda' \mathbf{v}$ belong to $C_3$, $\mathbf{0} \neq \mathbf{u} \in C_2$ and $\mathbf{0} \neq \mathbf{v} \in \hat{C}_2$.

\begin{rem}
{\rm
Notice that the Hamada's  generalization of the Steane's enlargement, Theorem \ref{ham} in this work, is a particular case of Theorem \ref{elnuevo} that holds when $C_1 = \hat{C}_1$.
}
\end{rem}

\section*{Acknowledgment}

The authors wish to thank Ryutaroh Matsumoto and the anonymous reviewers for helpful comments on this paper.


\begin{thebibliography}{99}
\footnotesize \setlength{\baselineskip}{3mm}

\bibitem{Akk} Aly, S.A., Klappenecker, A., Kumar, S., Sarvepalli, P.K. On quantum and classical BCH codes, {\it IEEE Trans. Inf. Theory} {\bf 53} (2007) 1183-1188.

\bibitem{AK} Ashikhmin, A., Knill, E. Non-binary quantum stabilizer codes, {\it IEEE Trans. Inf. Theory} {\bf 47} (2001) 3065-3072.

\bibitem{7kkk} Ashikhmin, A., Barg, A., Knill, E., Litsyn, S. Quantum error-detection I: Statement of the problem, {\it IEEE Trans. Inf. Theory} {\bf 46} (2000) 778-788.

\bibitem{8kkk} Ashikhmin, A., Barg, A., Knill, E., Litsyn, S. Quantum error-detection II: Bounds, {\it IEEE Trans. Inf. Theory} {\bf 46} (2000) 789-800.


\bibitem{BCM} Bian, Z. et al. Experimental determination of Ramsey numbers, {\it Phys. Rev. Lett.} {\bf 111} 130505 (2013).

\bibitem{BE} Bierbrauer, J., Edel, Y. Quantum twisted codes, {\it J. Comb. Designs} {\bf 8} (2000) 174-188.


\bibitem{maria-michael} Bras-Amor\'{o}s, M., O'Sullivan, M.E. Duality for some families of correction capability optimized evaluation codes, {\it Adv. Math. Commun.} {\bf 2} (2008) 15-33.

\bibitem{18kkk} Calderbank, A.R., Rains, E.M., Shor, P.W., Sloane, N.J.A. Quantum error correction and orthogonal geometry, {\it Phys. Rev. Lett.} {\bf 76} (1997) 405-409.

\bibitem{19kkk} Calderbank, A.R., Rains, E.M., Shor, P.W., Sloane, N.J.A. Quantum error correction via codes over GF(4), {\it IEEE Trans. Inf. Theory} {\bf 44} (1998) 1369-1387.

\bibitem{20kkk} Calderbank A.R., Shor, P. Good quantum error-correcting codes
exist, {\it Phys. Rev. A} {\bf 54} (1996) 1098-1105.


\bibitem{delsarte}   Delsarte, P. On subfield subcodes of modified Reed-Solomon codes, {\it IEEE Trans. Inform. Theory} {\bf IT-21} (1975) 575-576.

\bibitem{8AS} Dieks, D. Communication by EPR devices, {\it Phys. Rev. A} {\bf 92} (1982) 271.

\bibitem{eck} Ekert, A., Macchiavello, C. Quantum error correction for communication, {\it Phys. Rev. Lett.} {\bf 77} (1996) 2585.

\bibitem{edel} Edel, Y. \emph{Some good quantum twisted codes}.
Online available at {\tt http://\-www.\-mathi.\-uni-hei\-delberg.de\-/~yves/Matritzen/QTBCH/QTBCHIndex.html}.

\bibitem{ez} Ezerman, M.F., Jitman, S., Ling, S., Pasechnik. D.V. CSS-like constructions of asymmetric quantum codes, {\it IEEE Trans. Inf. Theory} {\bf 59} (2013) 6732-6754.

\bibitem{35kkk} Feng, K. Quantum error correcting codes. In Coding Theory and Cryptology, Word Scientific, 2002, 91-142.

\bibitem{feng} Feng, K., Ma, Z.  A finite Gilbert-Varshamov bound for pure stabilizer quantum codes, {\it IEEE Trans. Inf. Theory} {\bf 50} (2004) 3323-3325.


\bibitem{FL}Fitzgerald, J., Lax, R.F. Decoding affine variety codes using Gröbner bases,
{\it Des. Codes Cryptogr.} {\bf 13} (1998) 147-158.

\bibitem{galindo-hernando} Galindo, C., Hernando, F. Quantum codes from affine variety codes and their subfield subcodes.  {\it Des. Codes Crytogr.} {\bf 76} (2015) 89-100.

\bibitem{gal-her-rua} Galindo, C., Hernando, F., Ruano, D. New quantum codes from evaluation and matrix-product codes. Preprint arXiv:1406.0650.

\bibitem{galmon} Galindo C., Monserrat, F. Delta-sequences and evaluation codes defined by plane valuations at infinity, {\it Proc. London Math. Soc.} {\bf 98} (2009) 714-740.

\bibitem{galmon2} Galindo C., Monserrat, F. Evaluation codes defined by finite families of plane valuations at infinity. {\it Des. Codes Crytogr.} {\bf 70} (2014) 189-213.

\bibitem{Geil-Affine} Geil, O. {\it Evaluation codes from an affine variety code perspective}.  Advances in algebraic geometry codes, Ser. Coding Theory Cryptol. 5 (2008) 153-180. World Sci. Publ., Hackensack, NJ. Eds.: E. Martinez-Moro, C. Munuera, D. Ruano.

\bibitem{geil} Geil, O. Evaluation codes from order domain theory, {\it Finite Fields Appl.} {\bf 14} (2008) 92-123.

\bibitem{geil2} Geil, O., Matsumoto, R., Ruano, D. Feng-Rao decoding of primary codes, {\it Finite Fields Appl.} {\bf 23} (2013) 35-52.



\bibitem{38kkk} Gottesman, D. A class of quantum error-correcting codes saturating
the quantum Hamming bound, {\it Phys. Rev. A} {\bf 54} (1996) 1862-1868.

\bibitem{codet} Grassl, M. \emph{Bounds on the minimum distance of linear codes}. Online available at {\tt http://www.codetables.de}, accessed on 15th February 2015.

\bibitem{45kkk}  Grassl, M., R\"{o}tteler, M. Quantum BCH codes. In Proc. X Int. Symp. Theor.  elec. Eng. Germany 1999, 207-212.

\bibitem{opt} Grassl, M., Beth, T., R\"{o}tteler, M. On optimal quantum codes, {\it Int. J. Quantum Inform.} {\bf 2} (2004) 757-775.


\bibitem{ham} Hamada, M. Concatenated quantum codes constructible in polynomial time: Efficient decoding and error correction, {\it IEEE Trans. Inform. Theory} {\bf 54} (2008) 5689-5704.






\bibitem{jin} Jin, L., Ling, S., Luo, J., Xing, C. Application of classical Hermitian self-orthogonal MDS codes to quantum MDS codes, {\it IEEE Trans. Inform. Theory} {\bf 56} (2010) 4735-4740.



\bibitem{kkk} Ketkar, A.,  Klappenecker, A., Kumar, S., Sarvepalli, P.K. Nonbinary stabilizer codes over finite fields, {\it IEEE Trans. Inform. Theory} {\bf 52} (2006) 4892-4914.


\bibitem{lag3} La Guardia, G.G. Construction of new families of nonbinary quantum BCH codes, {\it Phys. Rev. A} {\bf 80} (2009) 042331.

\bibitem{lag2} La Guardia, G.G. On the construction of nonbinary quantum BCH codes, {\it IEEE Trans. Inform. Theory} {\bf 60} (2014) 1528-1535.

\bibitem{lag1} La Guardia, G.G., Palazzo, R. Constructions of new families of nonbinary CSS codes, {\it Discrete Math.} {\bf 310} (2010) 2935-2945.




\bibitem{magma} Magma Computational Algebra System. {\tt http://magma.maths.usyd.edu.au/magma/}.

\bibitem{Mar} Marcolla, C., Orsini, E., Sala M., Improved decoding of affine-variety codes, {\it J. Pure Appl. Algebra} {\bf 216} (2012) 147-158.

\bibitem{71kkk} Matsumoto, R., Uyematsu, T. Constructing quantum error correcting codes for $p^m$ state systems from classical error correcting codes. {\it IEICE Trans. Fund.} {\bf E83-A} (2000) 1878-1883.


\bibitem{mat} Matsumoto, R., Uyematsu, T. Lower bound for the quantum capacity of a discrete memoryless quantum channel, {\it J. Math. Phys} {\bf 43} (2002) 4391-4403.




\bibitem{Sarvepalli} Sarvepalli, P.K., Klappenecker, A.
 Nonbinary quantum Reed-Muller codes. In Proc. 2005 Int. Symp. Information Theory,  1023-1027.

\bibitem{sar1} Sarvepalli, P.K., Klappenecker, A., R\"{o}tteler, M. Asymmetric quantum codes: constructions, bounds and performance. {\it Proc. Royal Soc. A} {\bf 465} (2000) 1645-1672.

\bibitem{22RBC} Shor, P.W. Polynomial-time algorithms for prime factorization and discrete logarithms on a quantum computer, in Proc. 35th  ann. symp. found. comp. sc., {\it IEEE Comp. Soc. Press} 1994, 124-134.

\bibitem{23RBC} Shor, P.W. Scheme for reducing decoherence in quantum computer memory, {\it Phys. Rev. A} {\bf 52} (1995) 2493-2496.

\bibitem{3hagiw} Shor, P.W., Preskill, J. Simple proof of security of the BB84 quantum key distribution protocol, {\it Phys. Rev. Lett.} {\bf 85} (2000) 441-444.

\bibitem{SS} Smith, G., Smolin, J. Putting ``quantumness" to the test, {\it Physics} {\bf 6} 105 (2013).

\bibitem{95kkk} Steane, A.M. Simple quantum error correcting codes, {\it Phys. Rev. Lett.} {\bf 77} (1996) 793-797.

\bibitem{Steane-2} Steane, A.M. Multiple particle interference and  quantum error correction, {\it Proc. Roy. Soc. London A } {\bf 452} (1996) 2551-2577.


\bibitem{Steane-E} Steane, A.M. Enlargement of Calderbank-Shor-Steane quantum codes, {\it IEEE Trans. Inform. Theory } {\bf 45} (1999) 2492-2495.

\bibitem{26RBC} Wootters W.K., Zurek, W.H. A single quantum cannot be cloned, {\it Nature}
{\bf 299} (1982) 802-803.

\bibitem{f2} Yu, S. Bierbrauer, J., Dong, Y., Chen, Q., Oh, C.H. All the stabilizer codes of distance 3, {\it IEEE Trans. Inform. Theory } {\bf 59} (2013) 5179-5185.

\end{thebibliography}
\end{document}